\documentclass[12pt]{article}

\usepackage[margin=1in]{geometry}
\usepackage{setspace}

\usepackage{amsmath,amssymb,amsthm}
\usepackage{bm}
\usepackage{graphicx}

\newtheorem{theorem}{Theorem}
\newtheorem{lemma}{Lemma}

\newtheorem{proposition}{Proposition}

\newcommand{\Z}{\mathbb{Z}}
\newcommand{\N}{\mathbb{N}}

\newcommand{\R}{\mathbb{R}}

\newcommand{\Cov}{\mathrm{Cov}}
\newcommand{\Qt}{\widetilde{Q}}
\newcommand{\matern}{Mat\'ern}

\usepackage[round]{natbib}
\usepackage{authblk}

\title{Efficient Computation of Gaussian Likelihoods for Stationary Markov Random Field Models}

\author[1]{Joseph Guinness \thanks{jsguinne@ncsu.edu}}
\affil[1]{\it \small North Carolina State University, Department of Statistics}
\author[2]{Ilse C.\ F.\ Ipsen \thanks{ipsen@ncsu.edu}}
\affil[2]{\it \small North Carolina State University, Department of Mathematics}

\begin{document}

\maketitle

\abstract{\cite{rue2005gaussian} proposed a method for efficiently
  computing the Gaussian likelihood for stationary Markov random field
  models, when the data locations fall on a complete regular grid, and
  the model has no additive error term. The calculations rely on the
  availability of the covariances. We prove a theorem giving the rate
  of convergence of a spectral method of computing the covariances,
  establishing that the error decays faster than any polynomial in the
  size of the computing grid. We extend the exact likelihood
  calculations to the case of non-rectangular domains and missing
  values on the interior of the grid and to the case when an additive
  uncorrelated error term (nugget) is present in the model. We also
  give an alternative formulation of the likelihood that has a smaller
  memory burden, parts of which can be computed in parallel. We show
  in simulations that using the exact likelihood can give far better
  parameter estimates than using standard Markov random field
  approximations. Having access to the exact likelihood allows for
  model comparisons via likelihood ratios on large datasets, so as an
  application of the methods, we compare several state-of-the-art
  methods for large spatial datasets on an aerosol optical thickness
  dataset. We find that simple block independent likelihood and
  composite likelihood methods outperform stochastic partial
  differential equation approximations in terms of computation time
  and returning parameter estimates that nearly maximize the
  likelihood.}

\section{Introduction}

In a wide range of scientific and engineering applications, such as imaging, agricultural field trials, climate modeling, and remote sensing, spatial data are observed or reported on regular grids of locations, and the Gaussian process model is commonly used to model the data directly, or indirectly as a stage in a hierarchical model \citep{rue2009approximate,banerjee2014hierarchical}. Let $Y(\bm{x}) \in \R$, $\bm{x} \in \Z^2$ denote a real random variable indexed by the two-dimensional integer lattice. We say that $Y(\cdot)$ is a Gaussian process if, for any $n \in \N$ and any set of locations $\bm{x}_{1:n} = \{\bm{x}_1,\ldots,\bm{x}_n\}$, the vector $\bm{Y} = (Y(\bm{x}_1),\ldots,Y(\bm{x}_n))'$ has a multivariate normal distribution. 
A common decomposition of a Gaussian process is 
\begin{align*}
Y(\bm{x}) = \mu_\beta(\bm{x}) + Z(\bm{x}) + \varepsilon(\bm{x}),
\end{align*}
where $\mu_\beta(\bm{x})$ is a nonrandom function depending on mean parameter $\beta$, $Z(\cdot)$ is a Gaussian process with mean $E(Z(\bm{x})) = 0$ and covariance function $\Cov(Z(\bm{x}),Z(\bm{y})) = K_\theta(\bm{x},\bm{y})$ depending on covariance parameter $\theta$, and $\varepsilon(\bm{x}) \sim \mbox{i.i.d. } N(0,\sigma^2)$. The usual terminology in this scenario is to call $Z(\cdot)$ a latent process, since it is not directly observed, and to say that the model contains a nugget effect, referring to the presence of the $\varepsilon(\bm{x})$ term, which is used to model microscale variation or measurement error. 

We define the mean vector $\bm{\mu}_\beta := E(\bm{Y}) = (\mu_\beta(\bm{x}_1),\ldots,\mu_\beta(\bm{x}_n))'$, the random vector $\bm{Z} = (Z(\bm{x}_1),\ldots,Z(\bm{x}_n))'$ and its covariance matrix $\Sigma_\theta$. Then the covariance matrix for $\bm{Y}$ is $\Sigma_\theta + \sigma^2 I$, where $I$ is an identity matrix of appropriate size. The parameters can be estimated using the loglikelihood function for $\beta$, $\theta$, and $\sigma^2$ from $\bm{Y}$
\begin{align*}
L(\beta,\theta,\sigma^2) = -\frac{n}{2}\log(2\pi) - \frac{1}{2}\log\det (\Sigma_\theta + \sigma^2 I) - \frac{1}{2}(\bm{Y}-\bm{\mu}_\beta)'(\Sigma_\theta+\sigma^2 I)^{-1}(\bm{Y}-\bm{\mu}_\beta).
\end{align*}
Memory and time constraints begin to prohibit storing and factoring $\Sigma_\theta + \sigma^2 I$ when $n$ is between $10^4$ and $10^5$. Since the inclusion of a non-zero mean vector does not introduce a substantial additional computational burden, we assume $\bm{\mu}_\beta = 0$ in our discussion of the methods, to simplify the formulas and focus our attention on estimating covariance and nugget parameters.

There is significant interest in identifying flexible covariance functions $K_\theta(\bm{x},\bm{y})$ for which the likelihood function--or an approximation to it--can be computed efficiently \citep{sun2012geostatistics}. Gaussian Markov random field (GMRF) models \citep{rue2005gaussian} are of particular value because they provide a flexible class of models that induce sparsity in $\Sigma_\theta^{-1}$, a feature that can be exploited by sparse matrix factorization algorithms. A GMRF on the infinite integer lattice $\Z^2$ is naturally specified by the conditional expectation and variance of each observation given every other observation on $\Z^2$. Formally, we write
\begin{align}\label{condspecfull}
Z(\bm{x}) | \{ Z(\bm{y}) : \bm{y} \in \Z^2 \setminus \bm{x} \} \sim N \bigg( -\sum_{\bm{y} \in \Z^2 \setminus \bm{x} } \frac{\theta(\bm{x},\bm{y})}{\theta(\bm{x},\bm{x})} Z(\bm{y}) \, , \, 1/\theta(\bm{x},\bm{x}) \bigg),
\end{align}
where $\theta(\cdot,\cdot)$ is a nonrandom function that encodes the conditional specification of $Z(\cdot)$. We use the symbol $\theta$ for this function and for the vector of covariance parameters since they both control the covariance function of $Z(\cdot)$. If for every $\bm{x}$, $\theta(\bm{x},\bm{y})$ is zero for all but a finite number of locations $\bm{y}$, then the random field is said to be Markov. GMRFs can also be viewed as approximations to more general classes of Gaussian random fields \citep{rue2002fitting}. GMRFs are also used as a stage in other spatial hierarchical multiresolution models \citep{nychka2014multi} and non-Gaussian models \citep{rue2009approximate}. The most common approaches to achieving efficient Bayesian inference with latent GMRFs involve Markov chain Monte Carlo (MCMC) \citep{knorr2002block} or an integrated nested Laplace approximation (INLA) \citep{rue2009approximate} to make inferences about $\theta$ and $\sigma^2$. For large datasets, the usual MCMC methods involve repeated proposals of $\bm{Z}$ to integrate out the latent GMRF.

INLA uses a clever rewriting of the likelihood for the GMRF parameters given $\bm{Y}$ to avoid the need to integrate over the distribution of $\bm{Z}$. In principle, this same approach could be used in MCMC as well. However, when using a finite sample of observations to make inferences about parameters controlling processes defined on infinite lattices, edge effects are known to introduce nonneglible biases \citep{guyon1982parameter}. For GMRFs, the issue is that conditional distributions of observations on the boundary of the study region given only the observations on the interior cannot be easily expressed in terms of $\theta$ and depend on the particular configuration of the observations. This is also true for observations on the interior whose neighbors are missing. A common approach to mitigating edge effects is to expand the lattice and treat observations on the expanded lattice as missing \citep{paciorek2007bayesian,lindgren2011explicit,stroud2014bayesian,guinnesscirculant}. This approach is generally approximate in nature, and the practitioner must decide how much lattice expansion is required to mitigate the edge effects. \cite{stroud2014bayesian} describe an exact method, but \cite{guinnesscirculant} argue that the tradeoff for exactness is a large number of missing values that must be imputed, slowing the convergence of iterative algorithms. \cite{besag1995conditional} give an approximate method for mitigating edge effects based on modifying an approximation to the precision matrix when the data form a complete rectangular grid. \cite{dutta2014h} use a similar method paired with approximate likelihoods based on the h-likelihood \citep{henderson1959estimation}.

\cite{rue2005gaussian} proposed an exact method for efficiently computing the Gaussian likelihood for $\theta$ and $\sigma^2$ from $\bm{Y}$ when the random field stationary and has no nugget term, so when that method can be applied, there is no need to impute missing values or expand the lattice to a larger domain. The method does not impose any unrealistic boundary conditions on the spatial model--the likelihood is for the infinite lattice model observed on a finite lattice. This paper provides theoretical support for this method, extends it to a more general case, gives simulation and timing results for its implementation, and applies it to the problem of model comparison on a large spatial dataset. One drawback of the method is that it requires the covariances, which are not generally available in closed form. We prove a theorem establishing that a spectral method for computing the covariances converges faster than any polynomial in the size of the computing grid. Computing covariances with the spectral method uses fast Fourier transform (FFT) algorithms, and so is efficient in terms of memory and computational time. In addition, we extend the \cite{rue2005gaussian} to the case of a non-rectangular domain and missing values on the interior of the domain, which is a common scenario for satellite data, and to the case when the model has a nugget term, which often improves the fit of the models. We also give an alternative formulation of the likelihood that has a smaller memory burden and can be computed in parallel. 

In Section \ref{methodsection}, we review the likelihood calculation and outline our extensions to it. Section \ref{simulationsection} contains simulation results showing cases in which standard MRF likelihood approximations fail due to boundary effects, and timing results demonstrating the computational efficiency of our implementation of the methods in Matlab. We apply the methods to a set of aerosol optical thickness values observed over the Red Sea, which has an irregular coastline and several islands over which the data are not reported. The analysis of the Red Sea data is in Section \ref{datasection}, where we highlight the role of the exact likelihood for both parameter estimation and model comparison, and where stochastic partial differential equation approximations are outperformed by much simpler independent blocks likelihood approximations. We conclude with a discussion in Section \ref{discussionsection}.

\section{Efficient Gaussian Likelihood Calculations}\label{methodsection}
In this section, we outline the theoretical results underlying the likelihood calculations for GMRF models. For any $\bm{x} \in \Z^2$, define $S_\theta(\bm{x})$ to be the set $\{ \bm{y} \in \Z^2 : \theta(\bm{x},\bm{y}) \neq 0 \}$, which we refer to as the neighborhood set of $\bm{x}$ under $\theta$. If the random field is Markov, then it is possible that all neighbors of an observation location are contained in the finite set of all observation locations. 

\vskip11pt

\noindent {\bf Definition}: Let $1 \leq i \leq n$. If $S_\theta(\bm{x}_i) \subset \bm{x}_{1:n}$, we say that $\bm{x}_i$ is a {\it fully neighbored} location, and $Z(\bm{x}_i)$ is a {\it fully neighbored} observation. If a location or observation is not fully neighbored, we say that it is {\it partially neighbored}.

\vskip10pt

We define $m_n \leq n$ to be the number of partially neighbored observations among $\bm{x}_{1:n}$. The partially neighbored observations consist of those along the boundary of the study region and any observations on the interior of the study region that have at least one missing neighbor. Figure \ref{fullyneighbored} contains an illustration. For example, if $\bm{x}_{1:n}$ is a complete square grid of dimension $(1000,1000)$, and each location's neighborhood set is the four adjacent grid points, $m_n = 3996$, while $n = 10^6$. Practical application of the methods depends on the following lemma, so we state it here first to give context to the theoretical results that follow.

\begin{figure}
\centering
\includegraphics[width=0.7\textwidth]{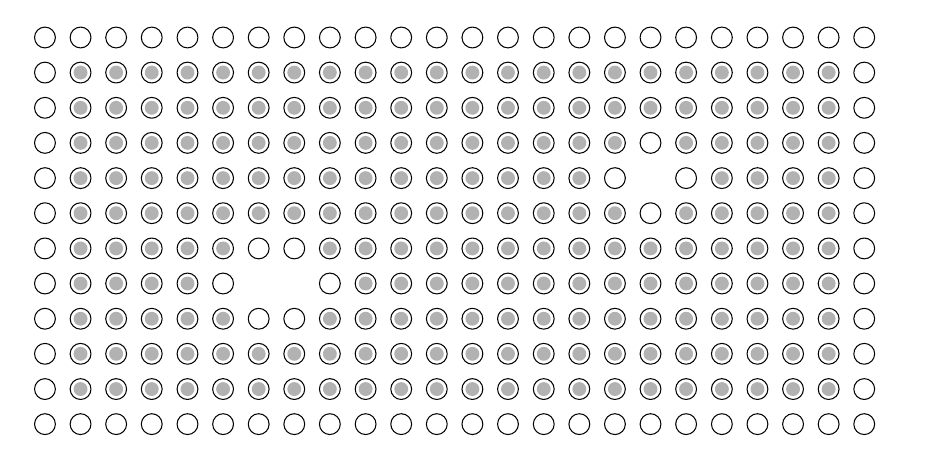}
\caption{Illustration of the fully and partially neighbored observations with neighborhood set consisting of four adjacent grid points. Black circles indicate observation locations, and those with a gray interior are the fully neighbored locations.}
\label{fullyneighbored}
\end{figure}

\begin{lemma}\label{eitherortheorem}
Let $\bm{x}_{1:n}$ be a finite set of observation locations for the infinite GMRF $Z(\cdot)$ specified by $\theta$, let $Q_\theta = \Sigma_\theta^{-1}$, the inverse of the covariance matrix for $\bm{Z}$, and let $\bm{x}_i$ and $\bm{x}_j$ be two observation locations. If either $\bm{x}_i$ or $\bm{x}_j$ is fully neighbored, then $Q_\theta[i, j] = \theta(\bm{x}_i,\bm{x}_j)$.
\end{lemma}

Let $P$ be a permutation matrix for which $P\bm{Z} = (\bm{Z}_1',\bm{Z}_2')'$ reorders the components of $\bm{Z}$ so that $\bm{Z}_1$ contains the $m_n$ partially neighbored observations and $\bm{Z}_2$ contains the $n-m_n$ fully neighbored observations. Then define the block matrices
\begin{align}\label{partition}
P\Sigma_\theta P' = \begin{bmatrix}
\Sigma_{11} & \Sigma_{12} \\
\Sigma_{21} & \Sigma_{22}
\end{bmatrix}, \quad 
P Q_\theta P' = \begin{bmatrix}
Q_{11} & Q_{12} \\
Q_{21} & Q_{22}
\end{bmatrix},
\end{align}
so that $\Sigma_{11}$ the covariance matrix for $\bm{Z}_1$, and $\Sigma_{22}$ is the covariance matrix for $\bm{Z}_2$. Since $Q_{12}$, $Q_{21}$ and $Q_{22}$ all contain either a row or a column corresponding to a fully neighbored observation, they are all sparse and have entries given by Lemma \ref{eitherortheorem}. On the other hand $Q_{11}$ is not guaranteed to be sparse, and its entries cannot be exactly determined by $\theta$. Indeed, this issue of edge effects has been a topic of continued discussion throughout the Markov random field literature, see for example \cite{besag1995conditional}, and \cite{besag1999bayesian}.

When $\theta(\bm{x},\bm{y})$ can be expressed as $\eta(\bm{x}-\bm{y})$ for every $\bm{x}$ and $\bm{y}$, the random field is said to be stationary. In this case, the covariance function has a spectral representation
\begin{align}\label{spectralrepresentation}
K_{\theta}(\bm{x},\bm{y}) = \int_{[0,2\pi]^2} \frac{\exp(i\bm{\omega}'(\bm{x}-\bm{y}))}{\sum_{\bm{h} \in \Z^2 } \eta(\bm{h})\exp(i\bm{\omega}'\bm{h})}d\bm{\omega},
\end{align}
where the spectral density $f_{\theta}(\bm{\omega}) = \left(\sum_{\bm{h} \in \Z^2 } \eta(\bm{h})\exp(i\bm{\omega}'\bm{h})\right)^{-1}$ is required to be nonnegative for all $\bm{\omega} \in [0,2\pi]^2$. Assume that all of the observation locations fall inside a rectangular grid of size $\bm{n} = (n_1,n_2)$, and the axes have been defined so that $n_1 < n_2$. To compute the covariances, we follow a suggestion by \cite{besag1995conditional} to numerically evaluate the integrals in \eqref{spectralrepresentation}. We use the sums
\begin{align}\label{riemanncalc}
K_\theta(\bm{x},\bm{y};\, J,\bm{n}) = \frac{4\pi^2}{n_1n_2 J^2} \sum_{\bm{j} \in F_J} f_\theta(\bm{\omega}_{\bm{j}})\exp(i\bm{\omega}_{\bm{j}}'(\bm{x}-\bm{y})),
\end{align}
where $\bm{\omega}_{\bm{j}} = 2\pi/J(j_1/n_1,j_2/n_2)$ are Fourier frequencies on a grid $F_J$ of size $(n_1J,n_2J)$. The following theorem gives a bound on the convergence rate of this method for calculating the covariances.

\begin{theorem}\label{covariancetheorem}
If $f(\bm{\omega})$ is the spectral density for a Markov random field on $\Z^2$ and is bounded above, then for $J \geq 2$ and any $p \in \N$, there exists a constant $C_p < \infty$ such that
\begin{align*}
\big| K_\theta(\bm{x},\bm{y}) - K_\theta(\bm{x},\bm{y};\,J,\bm{n})| \leq \frac{C_p} {(n_1J)^{p}}.
\end{align*}
\end{theorem}

\noindent We provide a proof in Appendix \ref{proofsection}. In other words, the error decays faster than a polynomial in $(n_1 J)$ of degree $p$ for every $p \geq 1$, so the sums converge quickly with $J$, and the error is smaller for a given $J$ with larger $n_1$. For large observation grids and common GMRF models, the covariances can be computed to an absolute error of less than $10^{-15}$ with $J$ as small as $2$ or $3$, and thus this method for computing covariances can be considered exact up to machine precision. We provide a numerical study in Appendix \ref{covcomputeappendix}. FFT algorithms allow us to compute covariances at all of the required lags $\bm{x} - \bm{y}$ in $O( (n_1 n_2 J^2)\log(n_1 n_2 J^2))$ floating point operations and $O(n_1 n_2 J^2)$ memory.

\cite{rue2005gaussian} proposed computing the likelihood for $\bm{Z}$ by first computing the likelihood for the points on the partially neighbored boundary of the domain, which requires computing their covariances, and then the conditional likelihood of the interior fully neighbored points given the partially neighbored points on the boundary. Here, we consider the more general case of having some partially neighbored points on the interior and a non-rectangular domain. The following proposition gives the general expressions for the determinant and inverse of $\Sigma_\theta$ that can be used when the model does not contain a nugget.

\begin{proposition}\label{nonuggettheorem}
If $\Sigma_\theta$ and $Q_\theta$ are partitioned as in \eqref{partition}, then
\begin{align*}
(a) \quad &\det( \Sigma_\theta ) = \det( \Sigma_{11} )/\det( Q_{22} )\\
(b) \quad &\Sigma_\theta^{-1} = 
\begin{bmatrix}
I & -\Sigma_{11}^{-1}\Sigma_{12}\\
0 & I
\end{bmatrix}
\begin{bmatrix}
\Sigma_{11}^{-1} & 0 \\
0 & Q_{22}
\end{bmatrix}
\begin{bmatrix}
I & 0 \\
-\Sigma_{21}\Sigma_{11}^{-1} & I
\end{bmatrix}
\end{align*}
\end{proposition}

\noindent Proposition \ref{nonuggettheorem} can be viewed as a matrix version of the decomposition of the likelihood $p(\bm{z}_1,\bm{z}_2) = p(\bm{z}_1)p(\bm{z}_2 | \bm{z}_1)$, since $\bm{Z}_1$ is normal with mean 0 and covariance matrix $\Sigma_{11}$, and $\bm{Z}_2 | \bm{Z}_1$ is normal with mean $\Sigma_{21}\Sigma_{11}^{-1}\bm{Z}_1$ and covariance matrix $Q_{22}^{-1}$, which is the Schur complement of $\Sigma_{11}$. The entries of $\Sigma_{11}$ are obtained with an FFT, as in Theorem \ref{covariancetheorem}, and the entries of $Q_{22}$ are known analytically. The various quantities required for evaluating the determinant and for multiplying $\Sigma_\theta^{-1}$ by $\bm{Y}$ can be easily computed after the Cholesky factorizations of $\Sigma_{11}$ and $Q_{22}$ have been obtained. Neither factorization is computationally demanding because $\Sigma_{11}$, though dense, is only of size $m_n \times m_n$, and $Q_{22}$, though large, is sparse and can be factored with sparse Cholesky algorithms. Multiplication of $\Sigma_{21}$ or $\Sigma_{12}$ by a vector is completed with circulant embedding techniques \citep{wood1994simulation}.
 
Including a nugget in the model increases the computational demands because $(\Sigma_\theta + \sigma^2 I)^{-1}$ is typically dense, even if $\Sigma_\theta^{-1}$ is sparse. Recent work on solving diagonally perturbed linear systems include \cite{Bella11,DuSog15,ErwMar12,SSX14,VKny13}. Here, we exploit the fact that most of the inverse of the unperturbed matrix is known and sparse, and we provide methods for evaluating the determinant in addition to solving linear systems. The covariance matrix for $\bm{Y}$ can be decomposed as
\begin{align}\label{nuggetdecomp}
\Sigma_\theta + \sigma^2 I= \Sigma_\theta ( I + \sigma^2 Q_\theta) = \Sigma_\theta ( \sigma^{-2} I + Q_\theta ) \sigma^2 I, 
\end{align}
and thus the following proposition holds.

\begin{proposition}\label{fullQtheorem}
\begin{align*}
(a) & \quad \det( \Sigma_\theta + \sigma^2 I) = \det(\Sigma_\theta) \det(I + \sigma^2 Q_\theta)\\
(b) & \quad (\Sigma_\theta + \sigma^2 I)^{-1} = (I + \sigma^2 Q_\theta)^{-1}Q_\theta.
\end{align*}
\end{proposition}
Equation \eqref{nuggetdecomp} and Proposition \ref{fullQtheorem} can be viewed as a matrix version of the decomposition of the likelihood $p(\bm{y}) = p(\bm{z})p(\bm{y} | \bm{z})/p(\bm{z} | \bm{y} )$, since the covariance matrix for $\bm{Z}$ is $\Sigma_\theta$, the covariance matrix for $\bm{Y} | \bm{Z}$ is $\sigma^2 I$, and the precision matrix for $\bm{Z} | \bm{Y}$ is $\sigma^{-2}I + Q_\theta$ \citep{lindgren2011explicit}. Exploiting this factorization requires one to know the dense matrix $Q_{11}$, but as a consequence of Theorem \ref{covariancetheorem}, $Q_{11}$ can be computed as
\begin{align*}
Q_{11} = \Sigma_{11}^{-1} + Q_{12}Q_{22}^{-1}Q_{21},
\end{align*}
which requires inversion of $\Sigma_{11}$, $m_n$ solves with the large sparse matrix $Q_{22}$, follwed by a multiplication with $Q_{12}$, which is also sparse. In practice, we have found that the $m_n$ solves are the most time-consuming step, but the solves can be completed in parallel.

After an appropriate reordering of the rows and columns, it is sometimes possible to compute and store the Cholesky factor of $I + \sigma^2 Q_\theta$ for the purpose of calculating its determinant and solving its linear system. Such reorderings typically place the $m_n$ partially neighbored observations last, and the resulting Cholesky factor is dense in the last $m_n$ rows. Thus the storage requirement is $O(n \times m_n)$. For large problems, this storage requirement can be prohibitively burdensome. The following proposition gives expressions for the determinant and inverse of $\Sigma_\theta + \sigma^2 I$ that can be used to evaluate the Gaussian loglikelihood function without storing any dense matrices larger than size $m_n \times m_n$.

\begin{proposition}\label{nuggettheorem}
Let $A := I + \sigma^2 Q_\theta$ and $B := A^{-1}$ be partitioned as in \eqref{partition}. Then
\begin{align*}
(a)& \quad \det( \Sigma_\theta + \sigma^2 I) = \det( \Sigma_\theta )\det(A_{22})\det( B_{11}^{-1} ) \\
(b)& \quad ( \Sigma_\theta + \sigma^2 I )^{-1} = 
\begin{bmatrix}
I & -\sigma^2 Q_{12}A_{22}^{-1} \\
0 & I
\end{bmatrix}
\begin{bmatrix}
B_{11} & 0 \\
0 & A_{22}^{-1}
\end{bmatrix}
\begin{bmatrix}
I & 0 \\
-\sigma^2 A_{22}^{-1}Q_{21}& I
\end{bmatrix}Q_\theta,
\end{align*}
and $B_{11}^{-1} = I + \sigma^2 Q_{11} - \sigma^4 Q_{12}A_{22}^{-1}Q_{21}$.
\end{proposition}

Since $B_{11}$ is the same size as $\Sigma_{11}$ and $A_{22}$ is as sparse as $Q_{22}$, most of the calculations are analogous to those that appear in Proposition \ref{nonuggettheorem}, so we omit most of the details here, noting however that multiplications with $Q_\theta$ are fast since it is mostly sparse, and we first form the matrix $B_{11}^{-1}$, which is the Schur complement of $A_{22}$, then factor it to compute its determinant and solve linear systems with it. Further, when forming $B_{11}^{-1}$, we do not need to store the entire matrix $A_{22}^{-1}Q_{21}$, which is dense and of size $n\times m_n$. Rather, we complete the solves involving each column $Q_{21}[:,j]$ individually, either in sequence or in parallel, only storing $Q_{12}A_{22}^{-1}Q_{21}[:,j]$ one at a time, without ever having to store the entire $n \times m_n$ matrix $A_{22}^{-1}Q_{21}$.

\subsection{Kriging and Conditional Simulations}

The inverse results in Propositions \ref{nonuggettheorem}, \ref{fullQtheorem}, and \ref{nuggettheorem} can also be used to perform Kriging of the data to unobserved locations, since the computationally limiting step in Kriging is solving a linear system involving the covariance matrix and the data. Let $\bm{Y}_0$
be a vector of unobserved values that we wish to predict using the vector of observations $\bm{Y}$. The joint distribution of the two vectors is
\begin{align*}
\begin{bmatrix}
\bm{Y} \\
\bm{Y}_0
\end{bmatrix}
\sim
N \bigg(
\begin{bmatrix}
\bm{\mu}\\
\bm{\mu}_0
\end{bmatrix}
,
\begin{bmatrix}
\Sigma_\theta + \sigma^2 I & \Sigma_0\\
\Sigma_0^T & \Sigma_{00} + \sigma^2 I
\end{bmatrix}
\bigg),
\end{align*}
where $\bm{\mu}$ and $\bm{\mu}_0$ are the two mean vectors, $\Sigma_0$ is the cross covariance matrix between $\bm{Y}$ and $\bm{Y}_0$, and $\Sigma_{00} + \sigma^2 I$ is the covariance matrix for the unobserved values. Then the conditional expectation of $\bm{Y}_0$ given $\bm{Y}$ is
\begin{align*}
E(\bm{Y}_0 | \bm{Y} ) = \bm{\mu}_0 + \Sigma_0^T(\Sigma_\theta + \sigma^2 I)^{-1}(\bm{Y}-\bm{\mu})
\end{align*}
The inverse results are used to solve the linear system with $\Sigma_\theta + \sigma^2 I$, and the forward multiplication with $\Sigma_0$ can be computed with circulant embedding techniques.

Conditional simulations of $\bm{Y}_0$ given $\bm{Y}$ do not require much additional computational effort. The conditional simulations consist of the conditional expectation (computed above), added to a simulation of a conditional residual with covariance matrix $\Sigma_{00} + \sigma^2 I - \Sigma_0^T(\Sigma_\theta + \sigma^2 I)^{-1}\Sigma_0$. The conditional residual can be formed with standard methods that are not more computationally demanding than the original Kriging computations, as long as one can generate unconditional simulations. See for example \citet[Section 7.3.1]{chiles2012geostatistics}. We use circulant embedding for the unconditional simulations.

\subsection{Existing Approximations to $Q_\theta$}\label{likelihoodapproximations}

We consider some existing likelihood approximations as competitors to the methods for computing the exact likelihood. The likelihood approximations replace $Q_\theta = \Sigma_\theta^{-1}$ with an approximation $\Qt_\theta$, with the various approximations differing in how $\Qt_\theta[i,j]$ is defined when both $\bm{x}_i$ and $\bm{x}_j$ are partially neighbored, that is, how $\widetilde{Q}_{11}$ is defined. Thus each approximation can be viewed as a method for dealing with edge effects and missing interior values. We may prefer to use a likelihood approximation if it is faster to compute and provides similar parameter estimates to the exact maximum likelihood parameters. We describe three approximations here and evaluate their effectiveness in Section \ref{simulationsection}.

\vskip12pt

\noindent {\it No Adjustment}: The simplest approximation defines $\Qt_\theta[i,j] = \theta(\bm{x}_i, \bm{x}_j)$ for every $i$ and $j$. \cite{rue2005gaussian} prove that $\Qt_\theta$ is positive definite when defined in this way.

\vskip12pt

\noindent {\it Precision Adjustment}: When the conditional specification satisfies the diagonal dominance criterion
\begin{align*}
\theta(\bm{x}_i,\bm{x}_i) = \lambda \sum_{\bm{x}_j \in \Z^2 \setminus \bm{x}_i} |\theta(\bm{x}_i,\bm{x}_j)|
\end{align*}
with $\lambda \in [0,1)$, we set $\Qt_\theta[i,j] = \theta(\bm{x}_i,\bm{x}_j)$ for $i\neq j$ and
\begin{align*}
\Qt_\theta[i,i] = \lambda \sum_{\bm{x}_j \in \bm{x}_{1:n} \setminus \bm{x}_i} |\theta(\bm{x}_i,\bm{x}_j)|
\end{align*}
when both $\bm{x}_i$ and $\bm{x}_j$ are partially neighbored. Then $\Qt_\theta$ is symmetric and diagonally dominant, and thus positive definite.

\vskip12pt

\noindent {\it Periodic (Toroidal) Adjustment}: Let $\bm{x}_{1:n}$ be a complete rectangular subset of the integer lattice of dimensions $\bm{n} = (n_1,n_2)$, so that $n = n_1 n_2$. Then we set 
\begin{align*}
\Qt_\theta[i,j] &= \theta( \bm{x}_i - \bm{x}_j + \bm{k} \circ \bm{n} ), \\
k_{\ell} &= \mbox{arg min}_{k \in \{-1,0,1\}} |x_{\ell i} - x_{\ell j} + k n_\ell|
\end{align*}
for $\ell = 1,2$, where $\bm{k} = (k_1,k_2)$, and $\bm{k} \circ \bm{n} = (k_1 n_1,k_2 n_2)$.

\section{Simulation and Timing Experiments}\label{simulationsection}

In this section we study the maximum likelihood and maximum approximate likelihood parameter estimates in simulation and timing experiments. For the simulations, we choose to consider the case of no nugget, so that measurement error does not obscure the importance of edge effects. This allows us to study how the approximate edge correction techniques perform compared to inference using the exact likelihood. Existing MCMC and INLA methods are designed for Bayesian inferene, and thus are not applicable here. We do, however, consider the performance of INLA in the data analysis in Section \ref{datasection}.

We consider spectral densities of the form
\begin{align}\label{qmaternspecden}
f_\theta(\bm{\omega}) = \tau^{-2}\big(\kappa^2 + 4 - e^{i\omega_1} - e^{-i\omega_1} - e^{i\omega_2} - e^{-i\omega_2} \big)^{-\nu -1}
\end{align}
where $\theta = (\tau,\kappa,\nu)$ with $\tau,\kappa > 0$ and $\nu = 0,1,2,\ldots$. \cite{lindgren2011explicit} showed that the spectral density in \eqref{qmaternspecden} is explicitly linked to the spectral density of the \matern\ covariance function with integer smoothness parameter $\nu$ and inverse range parameter $\kappa$. When $\nu = 0$, we have $\theta(\bm{0}) = \tau^2(\kappa^2+4)$, $\theta(\bm{h}) = -\tau^2$ when $\bm{h} = (1,0)$, $(-1,0)$, $(0,1)$, or $(0,-1)$, and $0$ otherwise. The $\nu=0$ model has been called a symmetric first order autoregression \citep{besag1995conditional}, although an alternative parameterization is more common. 

The following simulation studies demonstrate that edge effects adversely affect parameter estimation, especially when the spatial correlation is strong. We consider the cases of $\nu=0$ and $\nu=1$ with $\kappa \in \{1/5,1/10,1/20\}$ and  We conduct 100 simulations on a grid of size $(100,100)$ for each of the six parameter combinations. We compare the maximum approximate likelihood estimates to the maximum likelihood estimates, computed with our efficient methods. For the $\nu=0$ case, we find that the precision adjustment provides the best approximation. For the $\nu=1$ case, we find that the periodic adjustment provides the best approximation, although no adjustment performs similarly. The results are presented in Figure \ref{simresults}. In Appendix \ref{simulationexamplesection}, we provide a figure containing simulated fields with these six parameter combinations. The precision adjustment introduces a small bias when $\nu =0$. When $\nu=1$, the periodic adjustment performs poorly in every case, with the bias of the estimates increasing as $\kappa$ decreases. The maximum likelihood estimates are nearly unbiased in every case.

\begin{figure}
\centering
\includegraphics[width=1\textwidth]{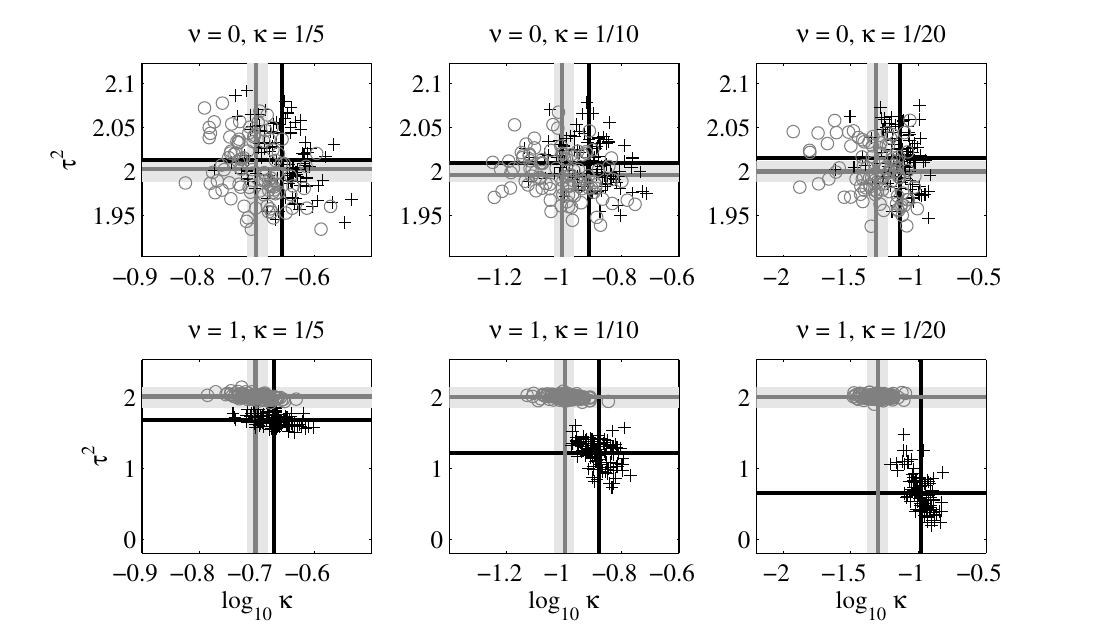}
\caption{Simulation results for maximum approximate likelihood estimates (black crosses) and maximum exact likelihood estimates (dark gray circles) for 100 simulated fields at each parameter combination. The thin solid lines indicate the sample means of the parameter estimates, with the sample mean of $\kappa$ taken on the log scale. The thick light gray lines indicate the true parameter values.}
\label{simresults}
\end{figure}

When there is no nugget in the model, the exact likelihoods not only produce more reliable parameter estimates, they often do not impose a significant additional computational burden over the calculation of the approximate likelihoods. In Table \ref{timingtable}, we present the results of a timing study for the $\nu=0$ and $\nu=1$ cases using no nugget and ``no adjustment'' neighborhood structure for the approximate likelihoods, which is the sparsest approximation. We compute the exact loglikelihood both with and without a small nugget. We use complete square lattices of increasing size, and time the computations using the ``clock'' and ``etime'' functions in Matlab. In both the $\nu = 0$ and the $\nu=1$ case with no nugget, the approximate likelihood is only 1.25 times faster than the exact
likelihood for the largest grid. Adding a nugget increases the time to compute the exact likelihood, but we note that the largest computations are infeasible without our methods, due to memory constraints on naive methods, and we have not parallelized the solves required for constructing $Q_{11}$. All computations are completed using Matlab R2013a on an Intel Core i7-4770 CPU (3.40GHz) with 32 GB of RAM.

\begin{table}
\centering
\begin{tabular}{ccccccc}
& \multicolumn{3}{c}{$\nu=0$} &\multicolumn{3}{c}{$\nu=1$} \\ 
& {Approximate} & Exact & Exact & Approximate & Exact & Exact\\
$n$ & $\sigma^2 = 0$ & $\sigma^2 = 0$ & $\sigma^2 = 0.01$ & $\sigma^2 = 0$ & $\sigma^2 = 0$ & $\sigma^2 = 0.01$ \\
\hline
$ 100^2 $ &   0.05 &   0.06 &   0.16 &   0.06 &   0.09 &   0.89 \\ 
$ 150^2 $ &   0.09 &   0.13 &   0.55 &   0.22 &   0.27 &   2.38 \\ 
$ 200^2 $ &   0.22 &   0.33 &   1.20 &   0.55 &   0.64 &   6.90 \\ 
$ 250^2 $ &   0.47 &   0.56 &   2.17 &   1.25 &   1.41 &  15.06 \\ 
$ 300^2 $ &   0.80 &   0.95 &   3.72 &   1.95 &   2.45 &  22.72 \end{tabular}
\caption{Time in seconds for one evaluation of the approximate likelihood and the exact likelihood, with and without a nugget term. }
\label{timingtable}
\end{table}

\section{Application to Method Comparison}\label{datasection}

The NASA Aqua satellite carries a moderate resolution imaging spectrometer (MODIS) capable of acquiring radiance data at high spatial resolution. The data are processed in order to obtain derived measurements of physical quantities relevant to land, ocean, and atmospheric dynamics, at various timescales. We analyze a set of aerosol optical thickness (AOT) data from Autumn 2014 over the Red Sea. The data were downloaded from NASA's OceanColor project website. Aerosol optical thickness is a unitless quantity describing the radiance attenuation caused by aerosols in the atmosphere. In Figure \ref{aotdata}, we plot the data, which consist of $21,921$ observations forming an incomplete regular grid. Even when gridded, satellite data often contain many missing values due to lack of coverage, atmospheric disturbances, or the quantity not being defined in certain regions (e.g. sea surface temperature over land).

The methods for computing the likelihood are useful both for estimating
parameters and for comparing models, allowing us to conduct model selection and to compare various methods for estimating parameters. We highlight both of these uses in the analysis of the AOT data. We model the data as
\begin{align*}
Y(\bm{x}) = \mu + Z(\bm{x}) + \varepsilon(\bm{x}),
\end{align*}
where $\mu \in \R$ is unknown and nonrandom, $Z(\cdot)$ is a GMRF with spectral density as in \eqref{qmaternspecden}, with $\nu = 1$, and unknown nonrandom parameters $\kappa,\tau > 0$, and $\varepsilon(\bm{x}) \sim \mbox{i.i.d. } N(0,\sigma^2)$. For this choice of $\nu = 1$, and this set of observation locations, we have $m_n = 3,469$. We estimate $\mu,\tau,\kappa$, and $\sigma^2$ by maximum likelihood, being careful to profile out $\mu$ and $\tau$, which have closed form expressions for the values that maximize the likelihood given fixed values of the other parameters. We also estimate $\mu$, $\tau$, and $\kappa$, fixing $\sigma^2 = 0$, allowing for a comparison between models with and without a nugget effect. Adding a nugget increased the loglikelihood by 594 units, indicating that the nugget significantly improves the model. The parameter estimates and loglikelihoods for these two models are given in Table \ref{dataanalysistable}.

\begin{figure}
\centering
\includegraphics[width=\textwidth]{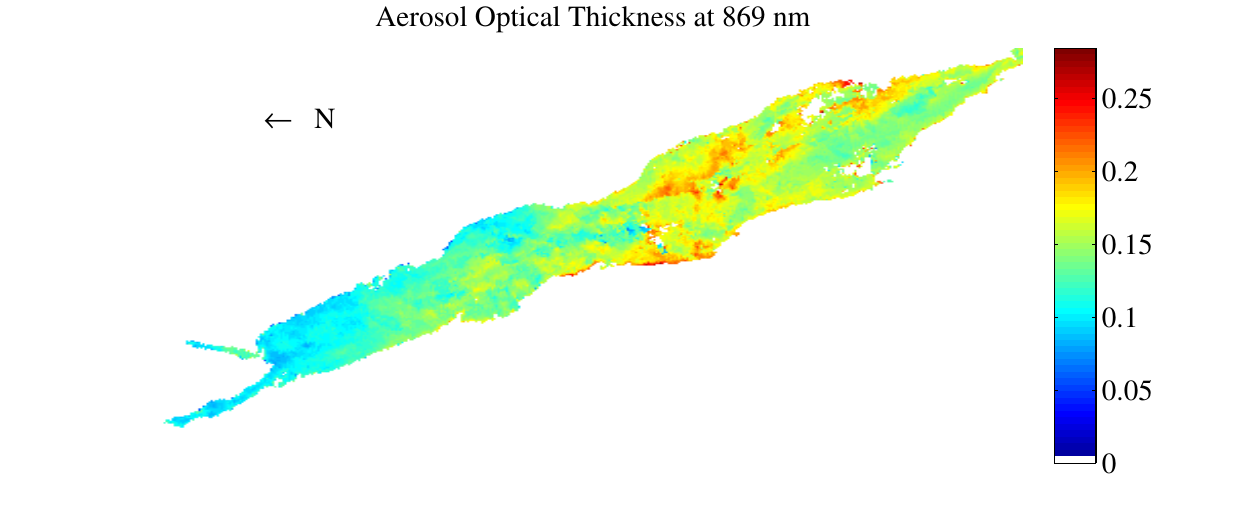}
\caption{Aerosol optical thickness in the Red Sea in Autumn 2014. White pixels are missing values.}
\label{aotdata}
\end{figure}

Having access to the exact likelihood allows for likelihood ratio comparisons, a convenient way to compare parameter estimates found using various approximate methods. We estimate parameters with two well-known likelihood approximations: the independent blocks likelihood approximation \citep{stein2014limitations}, and an approximation based on an
incomplete inverse Cholesky factorization \citep{vecchia1988estimation,stein2004approximating}. We also
estimate parameters using the R-INLA software (www.r-inla.org), described in \cite{rue2009approximate} and \cite{martins2013bayesian}, using a stochastic partial differential equation (SPDE) approximation \citep{lindgren2011explicit}. Since INLA is a Bayesian method, it does not return maximum likelihood estimates, but it is worthwhile to consider whether the INLA maximum posterior parameter estimates nearly maximize the likelihood, since INLA is designed specifically to estimate models with a latent GMRF.

There are several decisions that must be made when using the approximate methods. With the independent blocks method, we must choose the size and shape of the blocks. We partition the Red Sea according to the orientation in Figure \ref{partitionfigure}, so that each block has roughly 1600 observations, allowing the covariance matrices for the individual blocks to be stored and factored using standard dense methods. With the incomplete inverse Cholesky method, we choose smaller blocks, always using the adjacent block in the northwest direction as the conditioning set. In R-INLA, there are several decisions that must be made, including whether to force the nodes of the SPDE mesh to coincide with the observation locations, and how fine a mesh to use. And of course the priors must be chosen; we used the default priors. We did not enforce the mesh nodes to coincide with the grid. We tried values of 2, 4, and 6 for the ``max.edge'' option, with smaller values of max.edge giving finer meshes and better parameter estimates. 

\begin{figure}
\centering
\includegraphics[width=0.9\textwidth]{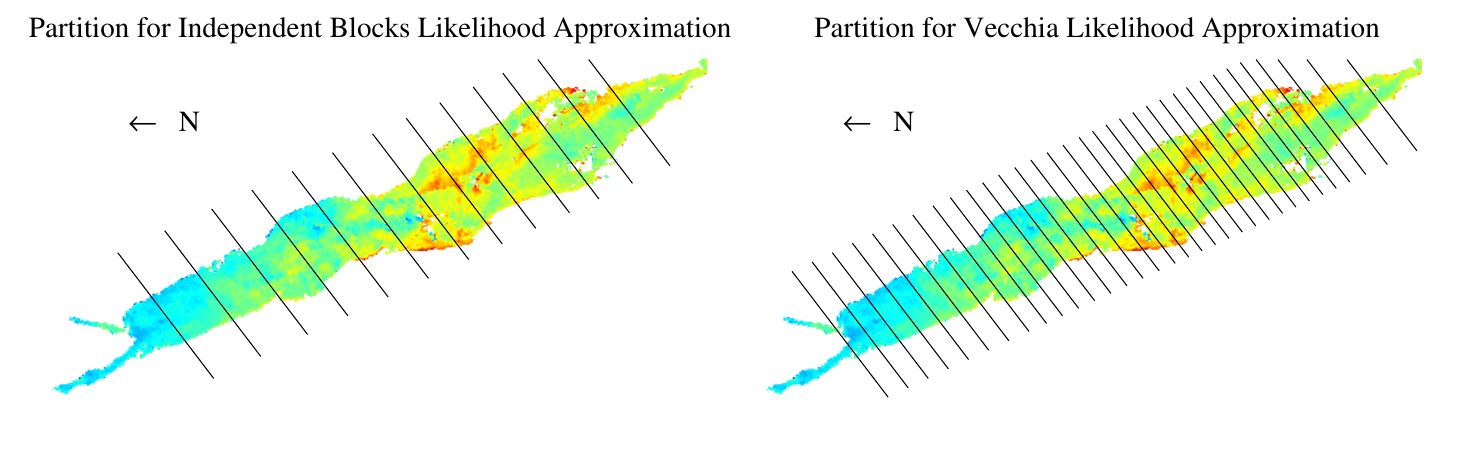}
\caption{Partitions used for the independent blocks approximation and incomplete inverse Cholesky approximation (Vecchia Approximation).}
\label{partitionfigure}
\end{figure}

The results of the model fitting are summarized in Table 2. The use of the exact likelihood to estimate the model delivers estimates in just 8.40 minutes, which is a major step forward for exact likelihood methods for a dataset of this size. The time required for INLA estimates depends on the resolution of the mesh, with finer meshes requiring more time. Estimating the parameters with the finest mesh took 33.2 minutes, nearly four times longer than generating the maximum likelihood estimates. The coarser meshes were faster than maximum likelihood, but the quality of the parameter estimates found with INLA depends on the mesh, with coarser meshes producing worse parameter estimates. The finest mesh that ran faster than maximum likelihood gave estimates whose loglikelihood is 795 units below the maximum loglikelihood, and decreasing the resolution further gave a loglikelihood more than 1300 units below the maximum loglikelihood, so it appears that the decisions made in constructing the mesh are very important. Comparisons with INLA should be made with caution, since it is a Bayesian method, but it is interesting to note that, in terms of loglikelihoods, the finest mesh is outperformed by the much simpler independent blocks approximation, which required much less computational effort. The incomplete inverse Cholesky approximation was better still and slightly slower than independent blocks. In INLA, we tried experimenting with different priors and did not see much dependence on the priors.

\begin{table}
\centering
\begin{tabular}{ccccccc}
method & $\mu$ & $\tau$ & $\kappa$ & $\sigma$ & $\Delta$loglik & time (min) \\
\hline
exact with nugget & 0.1464 &  51.25 &  0.2095 &  0.0048 & 0  & 8.40\\
inc.\ inv.\ Chol.\ & 0.1468 & 49.29 & 0.2094 & 0.0050 & $-33.4$ & 5.23\\
ind.\ blocks & 0.1485 & 49.12 & 0.2054 & 0.0051 & $-54.3$ & 4.15\\ 
INLA, max.edge = 2 & 0.1471 &  58.19 &  0.1715 &  0.0061 &$-158.9$& 33.2\\
exact no nugget & 0.1462 & 33.63 & 0.3648 & 0 & $-594.0$ & 0.41 \\
INLA, max.edge = 4 & 0.1490 &  79.02 &  0.1126 &  0.0077 &$-795.4$& 3.23\\
INLA, max.edge = 6 & 0.1497 &  88.49 &  0.0855 &  0.0085 &$-1344.1$ & 1.23
\end{tabular}
\caption{Results of analysis of Red Sea data. Loglikelihood differences are from maximum likelihood model.}
\label{dataanalysistable}
\end{table}

\section{Discussion}\label{discussionsection}

The Gaussian Markov random field, especially when paired with a nugget effect, is a powerful tool for modeling spatial data. This paper provides theoretica support for and generalizes a method proposed by \cite{rue2005gaussian} for computing the exact likelihood. The methods naturally handle edge effects, which we demonstrate can cause serious problems for estimating parameters if not treated carefully. The availability of the likelihood allows for both frequentist and Bayesian inference, and it allows for likelihood ratio comparisons among parameter estimates obtained using various approximate methods. The methods for computing the likelihood require forming the covariance matrix for the partially neighbored observations, and we show that the covariances can be computed efficiently and highly accurately with FFT algorithms. We make use of decompositions that exploit the fact that the inverse of the covariance matrix is mostly known and sparse. The likelihood calculations can be considered exact in the same sense that any calculation involving \matern\ covariances is considered exact, since evaluation of the \matern\ covariance function involves series expansions for the modified Bessel function.

The computational limits of the methods are determined by $m_n$, the number of partially neighbored observations, since we must construct and manipulate dense matrices of size $m_n \times m_n$. For complete square grids, this allows the computation of the likelihoods with more than one million observations. However, the methods are not able to handle situations in which $m_n$ is much larger than $10^4$, which arises when the observation grid has many scattered missing values. In this situation, using an iterative estimation method that relies on conditionally simulating the missing values given the data \citep{stroud2014bayesian,guinnesscirculant} is appropriate. Then the methods described in this paper can be used to calculate the likelihood for the entire imputed dataset.

\begin{center}
{\large \bf Acknowledgements}
\end{center}

The first author acknowledges support from the National Science Foundation's Division of Mathematical Sciences under grant number 1406016. The second author acknowledges the support from the XDATA Program of the Defense Advanced Research Projects Agency (DARPA), administered through Air Force Research Laboratory contract FA8750-12-C-0323 FA8750-12-C-0323.

\appendix

\section{Numerical Study}\label{covcomputeappendix}

We demonstrate that the stationary covariances for the random fields that we consider can be computed to high accuracy with small values of $J = 2$ or $3$. To study the calculations, for each value of $J \in \{1,2,3,4,5\}$, we compute the approximation in \eqref{riemanncalc}, and we record
\begin{align*}
\Delta_J = \max_{\bm{h}} \Big| K_\theta(\bm{x},\bm{x}+\bm{h}; \, J,\bm{n}) - K_\theta(\bm{x},\bm{x}+\bm{h}; \, J+1,\bm{n}) \Big|, 
\end{align*}
where the maximum absolute difference is taken over $\bm{h}$ on a grid of size $\bm{n} = (100,100)$, for $\nu \in \{0,1\}$, and $\kappa \in \{ 1/5,1/10,1/20 \}$, in order to represent a wide set of models. The results are plotted in Figure \ref{covcomputefigure}. We see that the calculations converge very rapidly, with $\Delta_3$ less than $10^{-10}$ for every parameter combination.

\begin{figure}[ht]
\centering
\includegraphics[width=0.7\textwidth]{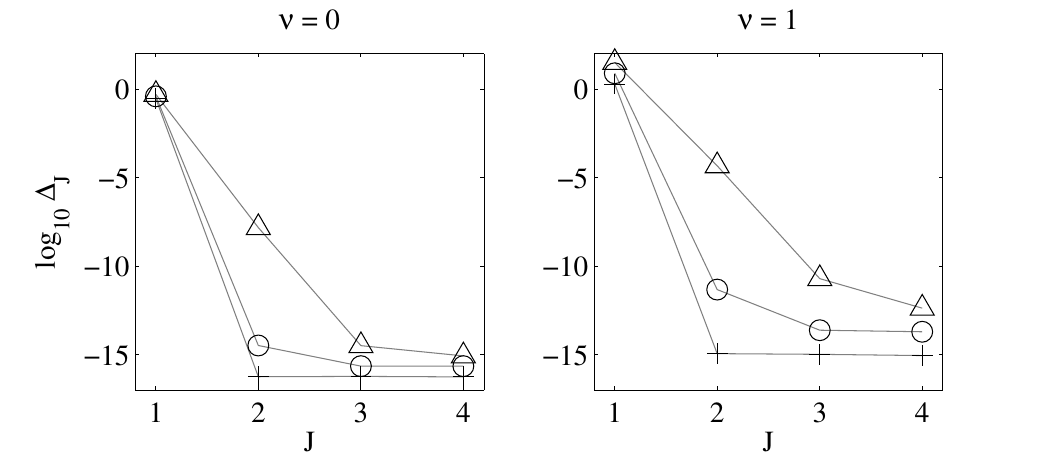}
\caption{$(+)$ represents $\kappa = 1/5$, $(\circ)$ represents $\kappa = 1/10$, and $(\triangle)$ represents $\kappa = 1/20$.}
\label{covcomputefigure}
\end{figure}

\section{Proofs}\label{proofsection}
\setcounter{theorem}{0}
\begin{theorem}
If $f(\bm{\omega})$ is the spectral density for a Markov random field on $\Z^2$ and is bounded above, then for $J \geq 2$ and any $p \in \N$, there exists a constant $C_p < \infty$ such that
\begin{align*}
\big| K_\theta(\bm{x},\bm{y}) - K_\theta(\bm{x},\bm{y};\,J)| \leq \frac{C_p} {(n_1J)^{p}}.
\end{align*}
\end{theorem}
\begin{proof}
Since the random field is Markov, $1/f(\bm{\omega})$ can be written as a finite sum of complex exponentials and must be infinitely differentiable, and thus $f(\bm{\omega})$ is also infinitely differentiable, since it is assumed to be bounded above. Since the process is stationary, we write $\bm{h} = \bm{x} - \bm{y}$, and $K(\bm{x},\bm{y}) = K(\bm{h})$ to simplify the expressions. \cite{guinnesscirculant} proved that
\begin{align*}
K( \bm{h};\, J,\bm{n} ) = \sum_{\bm{k} \in \Z^2} K( \bm{h} + (J \bm{n} \circ \bm{k})),
\end{align*}
where $\bm{n} \circ \bm{k} = (n_1k_1,n_2k_2)$. Thus 
\begin{align*}
K_\theta(\bm{h}) - K_\theta(\bm{h};\,J,\bm{n}) = \sum_{\bm{k} \neq \bm{0}} K( \bm{h} + (J \bm{n} \circ \bm{k})). 
\end{align*}
Since $f(\bm{\omega})$ is $p \geq 3$ times continuously differentiable, there exists constant $M_p < \infty$ such that $|K(\bm{r})| \leq M_p \| \bm{r} \|^{-p}$, where $\| \cdot \|$ is Euclidean distance \citep[Lemma 9.5]{korner1989fourier}. Then the error can be bounded by
\begin{align*}
|K_\theta(\bm{h}) - K_\theta(\bm{h};\,J,\bm{n})| &\leq \sum_{\bm{k} \neq \bm{0}} |K( \bm{h} + (J \bm{n} \circ \bm{k}))| \\
 & \leq 3 M_p \left( (J-1) n_1 \right)^{-p} + \sum_{j=1}^{\infty} 4(2j+1) M_p (jJn_1)^{-p}\\
 &= 3 M_p \left( (J-1)n_1 \right)^{-p} + 4M_p(Jn_1)^{-p}\sum_{j=1}^{\infty} (2j+1)(j)^{-p}\\
  & \leq 3 M_p \left( (J/2)n_1 \right)^{-p} + 4M_p(Jn_1)^{-p}\sum_{j=1}^{\infty} (2j+1)(j)^{-p}\\
  & \leq \left[ 3 M_p \left( 1/2 \right)^{-p} + 4M_p\sum_{j=1}^{\infty} (2j+1)(j)^{-p} \right] (Jn_1)^{-p}
\end{align*}
Setting $C_p$ equal to the quantity in square brackets establishes the bound for $p \geq 3$. The $p=1$ and $p=2$ cases follow by setting $C_1 = C_3$ and $C_2 = C_3$, since $(Jn_1)^{-3} < (Jn_1)^{-2} < (Jn_1)^{-1}$
\end{proof}

We believe the following are known results but were unable to find proofs, so we prove them here for completeness. Lemma 2 is used in the proof of Lemma 1.

\setcounter{lemma}{0}
\begin{lemma}
Let $\bm{x}_{1:n}$ be a finite set of observation locations for the infinite GMRF $Z(\cdot)$ specified by $\theta$, let $Q_\theta = \Sigma_\theta^{-1}$, the inverse of the covariance matrix for $\bm{Z}$, and let $\bm{x}_i$ and $\bm{x}_j$ be two observation locations. If either $\bm{x}_i$ or $\bm{x}_j$ is fully neighbored, then $Q_\theta[i, j] = \theta(\bm{x}_i,\bm{x}_j)$.
\end{lemma}
\begin{proof}
Using the form of the density function of the multivariate normal, the density function of $Z(\bm{x}_i)$ given $\bm{Z}_{-i}$ ($\bm{Z}$ with the $i$th element removed) is
\begin{align}\label{theorem1_1}
p(z(\bm{x}_i) | \bm{z}_{-i}) &\propto \exp \Big( -\frac{1}{2}\bm{z}'Q_\theta \bm{z} \Big) \notag\\
 & \propto \exp \bigg( -\frac{1}{2}z(\bm{x}_i)^2 Q_\theta[i,i] - z(\bm{x}_i)\sum_{j\neq i}Q_\theta[i,j]z(\bm{x}_j)  \bigg),
\end{align}
where $\bm{z} = (z(\bm{x}_1),\ldots,z(\bm{x}_n))$ is the argument for the joint density for $\bm{Z}$, and $\bm{z}_{-i}$ is $\bm{z}$ with the $i$th element removed. Suppose, without loss of generality, that $\bm{x}_i$ is fully neighbored. Then, using the conditional specification in \eqref{condspecfull} and the form of the univariate normal distribution, the density function of $Z(\bm{x}_i)$ given $\bm{Z}_{-i}$ is
\begin{align}\label{theorem1_2}
p(z(\bm{x}_i) | \bm{z}_{-i}) \propto \exp \bigg( -\frac{1}{2}z(\bm{x}_i)^2 \theta(\bm{x}_i,\bm{x}_i) - z(\bm{x}_i)\sum_{j\neq i}\theta(\bm{x}_i,\bm{x}_j)z(\bm{x}_j)  \bigg).
\end{align}
Comparing \eqref{theorem1_1} and \eqref{theorem1_2} proves that $Q_\theta[i,j] = \theta(\bm{x}_i,\bm{x}_j)$, since \eqref{theorem1_1} and \eqref{theorem1_2} must be equal for any $\bm{z}_{-i} \in \R^{n-1}$. The symmetry of $Q_\theta$ and $\theta$ (Lemma \ref{symmetriclemma}) implies that $Q_\theta[j,i] = Q_\theta[i,j] = \theta(\bm{x}_i,\bm{x}_j) = \theta(\bm{x}_j,\bm{x}_i)$, which establishes the stated result.
\end{proof}

\setcounter{lemma}{1}
\begin{lemma}\label{symmetriclemma}
If $\theta$ specifies a valid GMRF, $\theta(\bm{x}_i,\bm{x}_j) = \theta(\bm{x}_j,\bm{x}_i)$ for all $\bm{x}_i,\bm{x}_j \in \Z^2$.
\end{lemma}
\begin{proof}
Let $\bm{x}_{1:n}$ be a set of observation locations that includes $\bm{x}_i$, $\bm{x}_j$, $S_\theta(\bm{x}_i)$, and $S_\theta(\bm{x}_j)$. The conditional density of $Z(\bm{x}_i)$ given $\bm{Z}_{-i}$ is
\begin{align}\label{lemma1_1}
p( z(\bm{x}_i) | \bm{z}_{-i}) &\propto \exp \bigg( -\frac{1}{2}\bm{z}'Q_\theta \bm{z} \bigg) \notag \\
&\propto \exp \bigg( -\frac{1}{2} z(\bm{x}_i)^2Q_\theta[i,i] - z(\bm{x}_i) \sum_{j\neq i} Q_\theta[i,j] z(\bm{x}_j) \bigg).
\end{align}
Using the form of the univariate normal distribution and the conditional mean and variance of $Z(\bm{x}_i)$ given its neigbors, the conditional density of $Z(\bm{x}_i)$ given $\bm{x}_{1:n}$ is
\begin{align}\label{lemma1_2}
p( z(\bm{x}_i) | \bm{z}_{-i}) \propto \exp \bigg( -\frac{1}{2} z(\bm{x}_i)^2\theta(\bm{x}_i,\bm{x}_i) - z(\bm{x}_i) \sum_{j\neq i} \theta(\bm{x}_i,\bm{x}_j) z(\bm{x}_j) \bigg).
\end{align}
Comparing \eqref{lemma1_1} and \eqref{lemma1_2} and equating coefficients gives $Q_\theta[i,j] = \theta(\bm{x}_i,\bm{x}_j)$. Repeating the same steps for $\bm{x}_j$ gives $Q_\theta[j,i] = \theta(\bm{x}_j,\bm{x}_i)$. The symmetry of $Q_\theta$ implies that $\theta(\bm{x}_i,\bm{x}_j) = \theta(\bm{x}_j,\bm{x}_i)$.
\end{proof}

\section{Example Plots of Simulated Data}\label{simulationexamplesection}

Figure \ref{simexamples} shows realizations from the six models considered in the simulation study from Section \ref{simulationsection}.

\begin{figure}[ht]
\centering
\includegraphics[width=\textwidth]{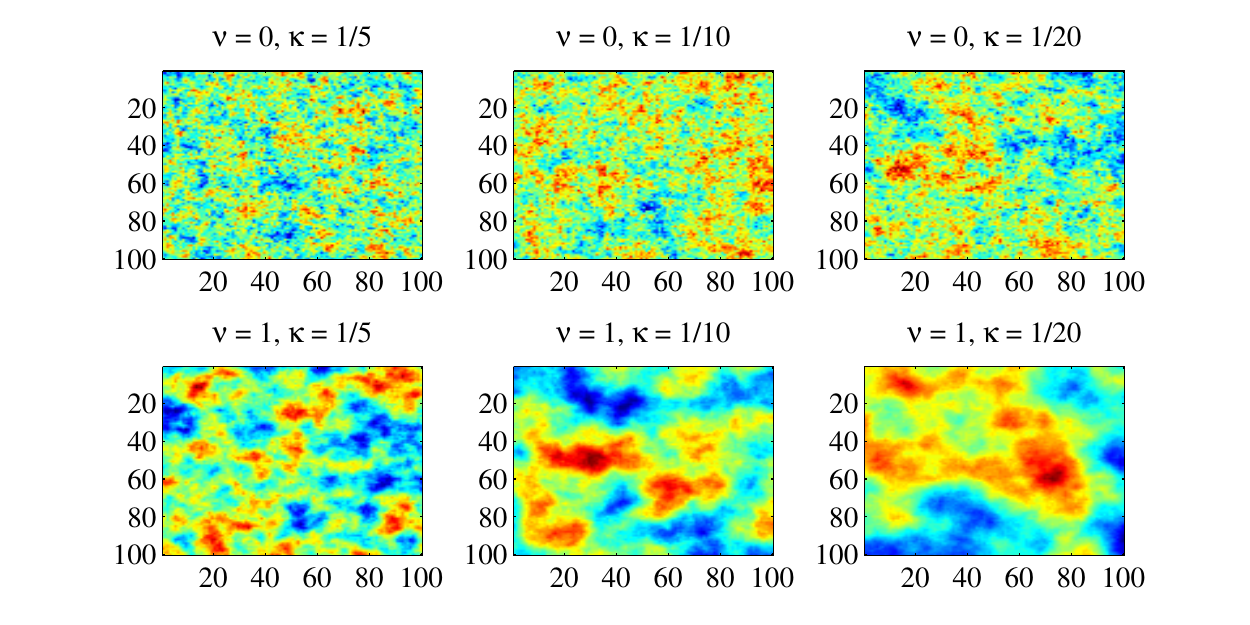}
\caption{Examples of simulated data for the six parameter combinations.}
\label{simexamples}
\end{figure}

\bibliographystyle{apa}
\bibliography{refs}

\end{document}